\documentclass[12pt]{article}
\usepackage[letterpaper, margin=1in]{geometry} 
\usepackage{graphicx}
\usepackage{hyperref}
\usepackage{amsmath}
\usepackage{amssymb}
\usepackage{nicefrac}
\usepackage{comment}
\usepackage{enumitem}
\usepackage{bm}
\usepackage{amsthm}
\theoremstyle{plain}
\newtheorem{theorem}{Theorem}[section]
\newtheorem{lemma}{Lemma}[section]      
\newtheorem{corollary}[theorem]{Corollary}  
\theoremstyle{definition}
\newtheorem{definition}{Definition}[section]  
\newtheorem{remark}{Remark}[section]          
\AtEndEnvironment{remark}{\hfill$\triangle$}
\usepackage{listings}
 
\title{Token Complexity Theory for AI-Augmented Computing}

\author{Jie Wang\footnote{Richard Miner School of Computing and Information Sciences, University of Massachusetts, Lowell, MA 01854, USA.} \\
jie\rule{0.4em}{0.4pt}wang@uml.edu}

\date{} 

\begin{document}

\maketitle

\begin{abstract}
AI-augmented computing delegates natural language queries, code 
generation requests, and other open-ended tasks to a cluster of AI 
models that processes queries and generates responses. This paradigm 
introduces a resource dimension that neither classical time nor space 
complexity captures: the cost of sending queries to and 
receiving responses from such a cluster. We introduce token 
complexity, a formal resource measure defined as the minimum expected 
token cost to achieve a specified level of output quality on a task, 
and develop a taxonomy classifying AI systems by the strength of 
their probabilistic properties.

We develop token complexity within the framework of AI-Oracle Turing 
machines, in which a probabilistic Turing machine interacts with a 
stochastic oracle via dedicated query and response tapes. We prove 
basic theorems establishing that token complexity behaves as expected: 
monotonicity (higher quality costs more tokens), convexity (quality 
improvements become progressively more expensive), price sensitivity 
(small price changes produce bounded cost changes), and 
price-relativity of task ordering (the token complexity ordering of 
tasks can reverse depending on the query-to-response cost ratio). We prove that the complexity frontier, defined as the set of all 
feasible resource bounds in tokens, time, and space, is non-empty, 
upward-closed, and convex.
\end{abstract}


\section{Introduction}\label{Introduction}

AI-augmented computing delegates open-ended tasks, including 
reasoning, parsing ambiguous text, generating code, evaluating 
arguments, summarizing documents, and producing digital artifacts, 
to AI systems via API calls. An AI system is a cluster of large language models (LLMs), 
in-house trained models, or other stochastic systems, with a control 
mechanism that routes each query to the best-suited model and returns 
model's response. 

Unlike a deterministic subroutine, an AI system returns a response 
drawn from a distribution rather than a fixed answer, introducing 
stochasticity into the computation. This 
delegation pattern is now the default paradigm for coding assistants, 
document processors, autonomous agents, and automated research workflows.

We formalize this pattern by treating an AI system as a stochastic oracle: a black box that accepts a query string and returns a response drawn from a distribution. We use \textit{oracle} as shorthand for stochastic oracle throughout, unless otherwise stated.

Interacting with stochastic oracles yields stochastic quality and incurs
costs of processing queries and generating responses, which differ 
from classical oracle Turing machines in three ways.

\begin{enumerate}
\item \textbf{Stochastic quality.} Classical deterministic oracles give the same answer to the same 
query string every time, with no notion of semantic equivalence 
between differently written queries of the same meaning. A 
stochastic oracle, by contrast, considers semantics through its 
response distributions: it may give similar responses to differently 
written queries of the same or closely related meaning, and different 
responses to queries that are syntactically similar but semantically 
distinct. Errors vary by query, correlate across calls, and are difficult 
to bound uniformly. Spending more query tokens may improve the probability of a 
higher-quality outcome, but does not guarantee it; response tokens 
are incurred passively as a cost of the oracle's answer.

\item \textbf{Oracle call cost.} Token cost is incurred precisely 
at the interface between the algorithmic component and the 
oracle---the boundary where probabilistic computation hands off to 
stochastic computation. This cost has two components: query tokens 
sent to the oracle and response tokens received from it. The total 
token cost separates raw token counts from their monetary cost, a 
distinction that proves essential: the relative cost of two tasks 
can reverse depending on the query-to-response cost ratio.

\textbf{Asymmetric cost.} A stochastic oracle typically incurs 
higher cost generating a re\-sponse than processing a query. Intuitively, 
processing a query is a frontend operation---interpreting and 
encoding the input---while generating a response is a backend 
operation involving reasoning and inference. Commercial LLMs make 
this cost explicit through differential pricing; open-source LLMs 
incur the same asymmetric cost implicitly through compute, without 
expressing it as a price. There is no classical analog of a 
complexity measure where the same token has a different cost 
depending on whether it is ingress or egress.
\end{enumerate}

These differences demand a new framework, one that not only measures the token cost of solving a given task but reveals how to reduce it, when reduction is possible, and when it is not.

The oracle boundary separates two fundamentally different kinds of computation. On the algorithmic side: probabilistic (including deterministic as a special case), mechanical, syntactic, carried out according to fixed rules and a random source. On the oracle side: stochastic, semantic, providing interpretation, inference, and reasoning that classical algorithms cannot perform without extensive task-specific engineering, if at all. The boundary is precisely where algorithmic computation suspends and oracle computation begins. Token complexity is therefore a measure of how much oracle computation a task requires, not asymptotically but precisely, as the minimum weighted token cost to achieve a given quality on a given task.

The practical stakes are concrete. Consider an application making one
million LLM calls per day, each sending 500 input tokens at USD 0.003
per thousand and receiving 200 output tokens at USD 0.012 per thousand.
At these costs, each output token costs four times more than each input
token. This asymmetry is typical of current commercial LLM API pricing, and where
it holds, the greater leverage lies in reducing response length rather
than query length. Understanding token complexity is therefore both theoretically 
necessary and practically consequential.\footnote{In a companion 
paper, we present case studies demonstrating AOTM designs that 
reduce token costs and design choices that depend on the 
query-to-response cost ratio.}

We assume familiarity with computational complexity theory at the 
level of Du and Ko~\cite{DuKo2001}, Arora and 
Barak~\cite{AroraBarak2009}, and Homer and 
Selman~\cite{HomerSelman2011}, and with AI systems such as large 
language models accessed via API calls.

The rest of this paper is organized as follows. Section~\ref{Computational 
Model} introduces the AI-Oracle Turing Machine model, in which a 
probabilistic Turing machine interacts with a stochastic oracle via a 
dedicated query-response interface. Section~\ref{Oracle Taxonomy} develops 
the oracle taxonomy, classifying oracles by the strength of their 
probabilistic properties and identifying Global Bounded Error as the aspirational target 
for LLM technology. Section~\ref{Token Complexity} defines token count 
and token cost, introduces token complexity as the minimum expected token 
cost to complete a task at a desired quality level, and presents the 
complexity frontier and tradeoff surface as the geometric objects 
characterizing the resource requirements of a task. 
Section~\ref{Properties of Token Complexity} shows that token complexity 
is monotone in quality, convex in quality, sensitive in price, and orders tasks in 
a way that depends on the query-to-response cost ratio. 
Section~\ref{Properties of the Complexity Frontier} shows that 
the complexity frontier is non-empty, upward-closed, and convex. Section~\ref{Open Problems} states central open problems and 
Section~\ref{Conclusion} concludes.

\section{AI-Augmented Computational Model}
\label{Computational Model}

We formulate the AI-augmented computational model in three steps. We
first introduce the AI-Oracle Turing Machine~\cite{Wang2025},
which extends a multi-tape probabilistic Turing machine with a dedicated
oracle interface. We then define the oracle as a family of response
distributions. Finally, we define the notion of a task, which
generalizes the classical decision problem and search problem as the
object of study for AI-augmented computation.

\subsection{AI-Oracle Turing Machines}

Throughout this paper, a fixed encoding scheme over a finite alphabet 
$\Sigma$ is assumed. All objects---strings, queries, responses, and 
mathematical objects---are represented as elements of $\Sigma^*$ in a 
standard way. Two strings that differ character by character are treated 
as distinct objects, regardless of semantic equivalence. This convention 
makes $\Sigma^*$ a well-defined universal domain for the framework. The 
length of a string $x \in \Sigma^*$, denoted $|x|$, is the number of 
symbols in $x$.

An AI-Oracle Turing Machine (AOTM) is a pair $\mathcal{M} = (M, 
\mathcal{O})$, where:

\begin{itemize}
  \item $M$ is a multi-tape probabilistic Turing machine (PTM) handling all 
  mechanical computation, consisting of an input tape, an output tape, 
  one or more work tapes, a read-only random tape containing an infinite 
  sequence of uniformly random bits, a query tape, and a response tape. 
  Deterministic Turing machines are the special case where the random 
  tape is never used. For simplicity, we omit the formal specification 
  of $M$'s finite set of states and transition function. For any given 
  task instance, $M$'s computation is determined solely by the instance 
  and its random tape; it does not have access to any data used to train 
  $\mathcal{O}$, and therefore cannot replicate the oracle's stochastic 
  capabilities. The PTM $M$ represents the \textit{algorithmic} side of $\mathcal{M}$.

  \item $\mathcal{O}$ is a stochastic oracle, formally defined as a 
family of response distributions $\mathcal{O} = \{\mathcal{D}_q : 
q \in \Sigma^*\}$, where $\mathcal{D}_q$ is the probability 
distribution over responses that $\mathcal{O}$ generates when given 
query $q \in \Sigma^*$. Access to $\mathcal{O}$ is exclusively 
through the dedicated query and response tapes. $\mathcal{O}$ 
represents the \textit{stochastic} side of $\mathcal{M}$.
\end{itemize}

\begin{remark}
This definition of oracle captures two practical abstractions of reality. 
First, the stochastic behavior of $\mathcal{O}$---its response 
distributions $\{\mathcal{D}_q\}$---is determined by the cluster 
of AI models used to solve a task, which may include LLMs, 
in-house trained models, or other stochastic systems. Second, the control mechanism reflects real AI 
infrastructure in practice: a load balancer, router, or 
orchestration layer directs queries to appropriate members based 
on task type, without exposing this routing to the calling 
application. Both are absorbed into the family $\{\mathcal{D}_q\}$.
\end{remark}

On each oracle call, $\mathcal{M}$ writes a query string $q \in \Sigma^*$ 
to the query tape, which is sent to $\mathcal{O}$. The oracle draws a 
response $r \in \Sigma^*$ from its distribution $\mathcal{D}_q$ and 
writes it to the response tape. $M$ reads the response on the response 
tape and continues its computation. Each oracle call is counted as a 
single computational step from $M$'s perspective: its computation simply 
suspends at each oracle call and resumes when the response arrives. The 
algorithmic side of the computation, namely $M$'s transitions, is 
measured in the standard way. The stochastic side of the computation is measured by the number of tokens crossing the oracle boundary and their associated costs. This 
division is formalized by the three computational resources of token 
complexity, time complexity, and space complexity defined in 
Section~\ref{Token Complexity}.

\begin{remark}
We assume that the time and space cost of an oracle call are bounded 
by a function of the query and response lengths---in practice 
near-linear for most deployed architectures---which are already 
captured by the token count. Under this assumption, the token cost 
is the dominant additional resource introduced by oracle interaction, 
beyond the time and space of $M$'s own computation. This assumption 
is well-supported empirically: for a deployed LLM with fixed 
architecture, computation time is bounded by a function of the total 
token length of the query and response, since all architecture 
parameters are fixed constants. For a cluster 
of AI models, the control mechanism's overhead depends on the number 
of models in the cluster; once the AOTM 
is fixed, the cluster size is fixed, making this overhead a constant 
relative to the query. The token count thus serves as a proxy for 
oracle time and space, further justifying it as the key resource 
measure for AI-augmented computing.
\end{remark}

\begin{remark}
The number of oracle calls, referred to as \textbf{turns} in practice, 
is not tracked as a separate resource: it is implicitly bounded above by 
the token count (since each call contributes at least one query token) 
and by the transition count (since initiating and receiving each call 
requires transitions of $M$). Studies focused on round complexity may 
wish to track turns explicitly as a fourth resource.
\end{remark}

\subsection{Tasks}\label{Tasks}

The object of computation for $\mathcal{M}$ is a \textbf{task}, defined 
as a tuple:
\[
T = (X, Y, S, \mathcal{D}_X),
\]
where $X \subseteq \Sigma^*$ is the input space, $Y \subseteq \Sigma^*$ 
the output space, $S: X \times Y \rightarrow [0,1]$ an explicitly 
defined score function, and $\mathcal{D}_X$ an input distribution. 
A specific input $x \in X$ 
is called a \textbf{task instance}.

The score function measures how well $\mathcal{M}$'s final output 
$y = \mathcal{M}(x) \in Y$ answers task instance $x$: $S(x, y) = 1$ 
represents a fully correct answer, $S(x, y) = 0$ represents complete 
failure, and intermediate values represent partially correct answers. The continuous 
range $[0,1]$ allows token complexity to capture the cost of achieving 
any quality level, not just the threshold between correct and incorrect. 

The score function takes one of four explicit forms depending on the 
task: verifier-based, where a program checks correctness against known 
ground truth; rubric-based, where a weighted rubric assigns a score; 
LLM-as-judge, where a language model scores the output, though 
self-critique variants are susceptible to correlated errors between 
generation and evaluation; or human annotation, where human preference 
defines correctness.

Two tasks that share the same $X$, $Y$, and score function but differ 
in $\mathcal{D}_X$ are formally distinct tasks with potentially 
different token complexities. In practice, $\mathcal{D}_X$ is 
approximated by the empirical distribution of observed inputs and need 
not be known exactly. It serves as a formal construct for defining 
expected token complexity, analogous to the role of input distributions 
in average-case complexity analysis~\cite{Levin1986, Wang1997}.

\begin{remark}
The score function $S$ is essential to the task definition: it is 
what makes the notion of solving a task well-defined. Without $S$, 
there is no criterion for what counts as a solution at quality level 
$\theta$, and token complexity---defined as the minimum token cost 
to achieve $\mathbb{E}[S(x, y)] \geq \theta$ for a given 
$\theta$---becomes trivial.

The score function $S$ is not part of the AOTM's 
definition; it is an external criterion applied to the final output, 
evaluating it against the task instance. $M$ may issue oracle 
queries whose responses feed back into its computation before 
producing the final output; how these are evaluated and used is 
determined by $M$'s transition function, not by $S$.
\end{remark}

\begin{remark}
The score function $S$, while not part of the AOTM's definition, 
is central to the design of an AOTM for solving a given task. Any 
machine $M$ designed to solve task $T$ must optimize for $S$: the 
querying strategy, oracle call patterns, and output construction 
are all chosen with the goal of maximizing 
$\mathbb{E}_{x \sim \mathcal{D}_X}[S(x, \mathcal{M}(x))]$. The 
separation between $S$ and the AOTM reflects the distinction 
between the evaluation criterion (what counts as a good answer) 
and the computational mechanism (how the answer is produced).
\end{remark}

\subsection{Conceptual Distinctions}\label{Notes on the Framework}

\textbf{Three-level structure.} The framework has three distinct levels. 
At Level 1, the task defines the problem to be solved and how solution 
quality is measured. At Level 2, the oracle is defined as a family of response distributions over query strings, independently of any task or machine design. 
At Level 3, the AOTM 
transforms a task instance into a solution according to $M$'s transition 
function, possibly querying the oracle in the process. Prompt 
engineering, chain-of-thought decomposition, verification loops, and 
other interaction patterns all operate at Level 3. Token complexity is 
a property of the Level 3 computation, not of the oracle or task alone. 
This separation ensures that oracle behavior and machine computation are 
characterized independently.

\textbf{Syntax, not semantics.} All objects in the framework---task 
instances, queries, responses, and outputs---are identified by their 
string representation, not by their semantic content. Two strings 
are distinct if and only if they differ character by character, 
regardless of whether they express the same underlying meaning. The 
oracle may assign different responses to two syntactically distinct 
but semantically equivalent queries. This is a deliberate choice: 
the oracle's response depends on the query string, not on its 
meaning. Semantics are not absent from the framework, however; they 
are captured through the oracle's response distributions 
$\{\mathcal{D}_q\}$, which may assign similar distributions to 
queries of similar meaning.

\textbf{Task instances and queries are distinct objects.} A task 
instance is the problem to be solved; it is the input against which the 
final output of $\mathcal{M} = (M, \mathcal{O})$ is evaluated by the 
score function. A query is a string $M$ sends to the oracle during its 
computation, generated according to $M$'s transition function. The same 
instance can generate many different queries depending on how $M$ 
chooses to ask: it may send the task verbatim, rephrase it, decompose 
it into sub-questions, add instructions or context, or construct 
something entirely different. This separation is intentional: it allows 
query formulation strategies to be studied independently of the task 
definition.

\section{Oracle Taxonomy}\label{Oracle Taxonomy}

The oracle taxonomy classifies oracles independently of any AOTM or 
task, based on conditions on their response distributions, assessed 
through quality measures applied to query-response pairs. The taxonomy 
(see Figure~\ref{figure1}) consists of six conditions: Measurable 
(baseline), Guaranteed Quality (GQ), Fast Error Decay (FED), Distinct 
Expected Quality (DEQ), Bounded Error (BE), and Global Bounded Error 
(GBE). GBE implies all others; BE implies Measurable, GQ, and FED but 
not DEQ; GQ, FED, and DEQ are pairwise non-implying; and Measurable is 
implied by all other conditions. These structural relationships are 
established formally in 
Section~\ref{Structural Relations Between Conditions}.

\begin{figure}[!htbp]
\centering
\includegraphics[width=0.7\linewidth]{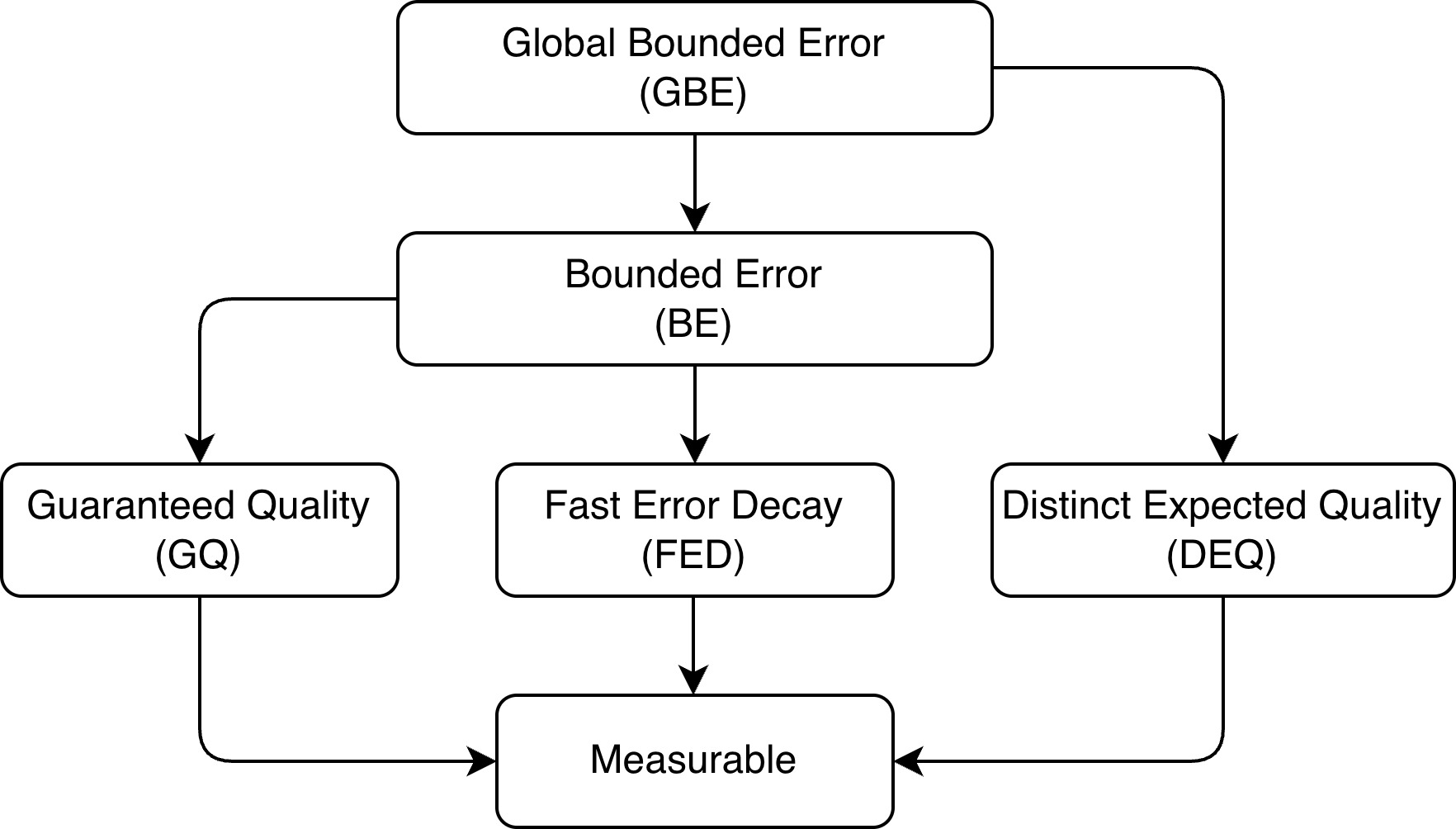}
\caption{Oracle taxonomy, where an arrow from A to B means A implies B, namely, every oracle satisfying A also satisfies B.}
\label{figure1}
\end{figure}

Throughout this paper, \(\delta_r\) denotes the Dirac measure at \(r\);
that is, the probability distribution on \(\Sigma^*\)\textit{ }that
assigns mass 1 to \(r \in \Sigma^*\) and mass 0 to every other string in
\(\Sigma^*\).

\subsection{Oracle Conditions}\label{Oracle Conditions}

For each condition we explain what it means, why it matters for token 
complexity, and how it relates to the other conditions in the taxonomy.

\begin{definition}[Measurable]
An oracle $\mathcal{O}$ is \textbf{Measurable} if for every query 
$q \in \Sigma^*$, $\mathcal{D}_q$ is a well-defined probability measure 
on $\Sigma^*$.
\end{definition}

\textit{What it means:} The oracle has a genuine probability distribution 
over responses for each query string. Since $\Sigma^*$ is countable, 
every oracle in this framework is automatically measurable in the 
measure-theoretic sense; Measurable is stated explicitly to establish 
the mathematical baseline and to support the logical structure of the 
taxonomy.

\textit{Why it matters:} This is the weakest useful condition. It rules 
out pathological oracles whose behavior has no well-defined probabilistic 
description, for which expected quality and token cost are undefined.

\begin{definition}[Guaranteed Quality]
An oracle $\mathcal{O}$ satisfies \textbf{Guaranteed Quality (GQ)} with 
respect to a quality measure $S: \Sigma^* \times \Sigma^* \rightarrow 
[0,1]$ if there exists a constant $\mu_0 \in (0,1]$ such that 
$\mathbb{E}_{r \sim \mathcal{D}_q}[S(q, r)] \geq \mu_0$ for all 
$q \in \Sigma^*$.
\end{definition}

\textit{What it means:} For the given quality measure $S$, the oracle 
guarantees a minimum expected quality $\mu_0 > 0$ on every query string, 
uniformly across all queries.

\textit{Why it matters:} This is the weakest condition guaranteeing that 
the oracle is useful: it ensures that expected quality is bounded away 
from zero on every query, which is sufficient to define a non-trivial 
quality constraint and guarantee the feasibility of token complexity at 
some positive threshold.

\begin{definition}[Fast Error Decay]
An oracle $\mathcal{O}$ satisfies \textbf{Fast Error Decay (FED)} with 
respect to a quality measure $S: \Sigma^* \times \Sigma^* \rightarrow 
[0,1]$ if there exist a decay rate $\lambda > 0$ and a constant $C > 0$ 
such that for any query $q \in \Sigma^*$ and any two calls $i \neq j$:
\[
\bigl|\mathrm{Cov}_{r_i, r_j \sim \mathcal{D}_q}
\bigl[S(q, r_i),\, S(q, r_j)\bigr]\bigr| \leq C \cdot e^{-\lambda m}
\]
where $m = |j - i| - 1$ is the number of intervening oracle calls.
\end{definition}

\textit{What it means:} Under the given quality measure $S$, repeated 
calls to the oracle on the same query string produce correlated scores: 
the covariance between scores on nearby calls decays exponentially with 
their separation, at a rate governed by the correlation time 
$t_\text{corr} = 1/\lambda$. After $k = t_\text{corr}$ intervening 
calls the covariance has decayed by a factor of $1/e$, and after 
$k = n \cdot t_\text{corr}$ intervening calls by a factor of $e^{-n}$. 
A small $t_\text{corr}$ means scores decorrelate quickly; a large 
$t_\text{corr}$ means scores remain significantly correlated across many 
calls. The limiting case $\lambda \to \infty$ (equivalently 
$t_\text{corr} \to 0$) corresponds to calls that are asymptotically 
uncorrelated at any separation $m \geq 1$.

\textit{Why it matters:} A large $t_\text{corr}$ limits the benefit of 
repeated querying: correlated calls cannot be treated as independent 
samples---sending the same query ten times gives ten correlated answers, 
not ten independent ones---so repeated sampling does not reduce error. 
FED is a condition on the oracle's stateless response mechanism; a 
stateful application such as a chatbot that maintains conversational 
history may violate FED, as errors can persist and compound across turns.

\begin{remark}
FED bounds the rate at which the covariance between two calls on the 
same query decays as their separation $|j - i|$ increases. It says 
nothing about the expected quality of individual calls, which depends 
on $\mathcal{D}_q$ and $S$, nor about any trend, improvement, or 
ordering in quality across calls. As $|j - i| \to \infty$, the two 
calls become asymptotically uncorrelated---but not necessarily 
independent, since zero covariance does not imply independence in 
general.
\end{remark}

\begin{definition}[Distinct Expected Quality]
An oracle $\mathcal{O}$ satisfies \textbf{Distinct Expected Quality 
(DEQ)} with respect to a quality measure $S: \Sigma^* \times \Sigma^* 
\rightarrow [0,1]$ if for any two queries $q, q' \in \Sigma^*$, 
$\mathcal{D}_q \neq \mathcal{D}_{q'}$ implies $\mathbb{E}_{r \sim 
\mathcal{D}_q}[S(q, r)] \neq \mathbb{E}_{r' \sim \mathcal{D}_{q'}}
[S(q', r')]$.
\end{definition}

\textit{What it means:} Distinct response distributions produce distinct 
expected quality levels under $S$: the map from response distributions 
to expected quality values is injective. DEQ ensures that oracle quality 
is observably distinguishable across queries---if the oracle responds 
differently to two queries, this difference is reflected in their 
expected quality levels.

\textit{Why it matters:} DEQ enables adaptive token optimization: 
without it, two queries could have distinct response distributions 
yet identical expected quality levels, making it impossible to 
distinguish their quality from observed responses alone and 
therefore impossible to optimize query formulation based on quality 
feedback. With DEQ, estimating $\mathbb{E}_{r \sim \mathcal{D}_q}
[S(q, r)]$ from observed responses becomes possible, enabling the 
machine to adapt its querying strategy: spending more tokens on 
low-quality queries, routing them differently, or applying 
verification steps.

The opposite of DEQ---response quality independent of response 
distributions---is desirable for deployment stability: an oracle 
that achieves the same expected quality regardless of which internal 
model handles a query is robust to model swaps and prompt 
variations. DEQ and its opposite thus serve different goals: DEQ 
enables optimization and calibration, while its opposite enables 
robustness and predictability. The taxonomy includes DEQ because 
the paper's focus is on optimizing token use, for which 
distinguishability of quality is essential.

\begin{definition}[BE]
An oracle $\mathcal{O}$ satisfies \textbf{Bounded Error (BE)} with 
respect to a binary quality measure $S: \Sigma^* \times \Sigma^* 
\rightarrow \{0, 1\}$ if for every query $q \in \Sigma^*$:
\begin{enumerate}[noitemsep, topsep=5pt]
    \item $\mathbb{E}_{r \sim \mathcal{D}_q}[S(q, r)] \geq 1/2 + \eta$ 
    for some fixed $\eta > 0$ (bounded-error success);
    \item repeated calls on the same query $q$ produce independent draws 
    from $\mathcal{D}_q$ (independence across calls).
\end{enumerate}
\end{definition}

\textit{What it means:} The oracle answers every binary task instance 
correctly with probability at least $1/2 + \eta$ for some fixed 
$\eta > 0$, with repeated calls that are independent across repetitions.

\textit{Why it matters:} BE oracles admit probability amplification for 
success: since calls are independent and success probability is bounded 
away from $\tfrac{1}{2}$, majority voting over $n$ independent calls 
achieves error at most $e^{-2\eta^2 n}$ by the Chernoff 
bound~\cite{Gill1977}, driving success probability arbitrarily close 
to $1$ exponentially in $n$.

\begin{remark}
The BE condition is the oracle analog of the 
probabilistic component of the classical BPP complexity class 
\cite{Gill1977}, which classifies problems solvable by a 
probabilistic polynomial-time Turing machine with bounded-error 
success probability. BE retains the two probabilistic properties 
that are meaningful for a black-box oracle: bounded-error success 
probability and independence of repeated calls.
\end{remark}

\begin{definition}[Global BE]
An oracle $\mathcal{O}$ satisfies \textbf{Global Bounded Error (GBE)} 
with respect to a binary quality measure $S: \Sigma^* \times \Sigma^* 
\rightarrow \{0,1\}$ if it satisfies both BE and DEQ with respect to $S$.
\end{definition}

\textit{What it means:} GBE combines probability amplification for 
success (BE) with the guarantee that oracle quality is distinguishable 
across queries (DEQ): distinct response distributions produce distinct 
expected quality levels, and repeated calls can be amplified to drive 
success probability arbitrarily close to $1$. Since BE implies 
Measurable, GQ, and FED but not DEQ (see 
Section~\ref{Structural Relations Between Conditions}), a GBE 
oracle---which additionally requires DEQ---satisfies all five conditions 
in the taxonomy.

\textit{Why it matters:} GBE is the aspirational target for LLM 
technology. BE alone does not imply DEQ---an oracle could have uniform 
success probability across all queries, making it impossible to 
distinguish oracle quality across different query strings. GBE rules 
this out: it requires both the probability amplification for success 
of BE and the ability to distinguish oracle quality across queries.

Among these six conditions, Measurable and GQ are pointwise conditions, 
each requiring every response distribution $\mathcal{D}_q$ to satisfy 
a property independently of the others. FED is a relational condition: 
it constrains the rate at which score covariance decays in magnitude 
across repeated calls on the same query. DEQ is a global condition: 
distinct queries must produce distinct expected quality levels. BE is 
simultaneously pointwise and relational. GBE is the strongest condition 
in the taxonomy, being simultaneously pointwise, relational, and global.

\begin{remark}[Lipschitz condition]
Lipschitz smoothness of the map $q \mapsto \mathcal{D}_q$ under an 
embedding $\varphi: \Sigma^* \rightarrow V$, where $V$ is a metric 
space, is another natural condition on stochastic oracles. It requires:
\[
d_{\mathrm{TV}}(\mathcal{D}_q, \mathcal{D}_{q'}) \leq \lambda \cdot 
d_V(\varphi(q), \varphi(q')) \quad \text{for all } q, q' \in \Sigma^*,
\]
where $\lambda > 0$ is a Lipschitz constant, $d_V$ is any metric on $V$, 
and $d_{\mathrm{TV}}$ is the total variation distance. We exclude it 
from the taxonomy because the right class of embeddings for AI oracle 
analysis remains an open question, as does the relationship between 
Lipschitz smoothness and the other conditions in the taxonomy under 
natural embeddings.

Empirically, LLMs appear to fail Lipschitz smoothness: small changes 
in query wording can produce large shifts in response distribution, 
consistent with an unbounded Lipschitz constant. A formal witness would 
require specifying $\varphi$ and exhibiting a concrete pair $(q, q')$ 
that violates the bound, which we leave as an open problem.
\end{remark}

\subsection{Structural Relations Between Oracle Conditions}
\label{Structural Relations Between Conditions}

The following lemmas establish the logical relationships between the 
six oracle conditions, justifying the taxonomy structure shown in 
Figure~\ref{figure1}. The Measurable condition is the baseline: it is 
implied by all other conditions, since each requires well-defined 
response distributions $\mathcal{D}_q$ to be meaningful.

\begin{lemma}[Measurable is necessary for all conditions]
Every oracle satisfying any of GQ, FED, DEQ, BE, or GBE also satisfies 
Measurable.
\end{lemma}

In all constructions below, each $\mathcal{D}_q$ is defined as a 
mixture of Dirac measures on $\Sigma^*$, so $\mathcal{O}$ is Measurable 
in every case. We do not repeat this argument for each construction.

\begin{lemma}[Pairwise non-implication of intermediate conditions]
Among GQ, FED, and DEQ, no condition implies any other with respect 
to the same quality measure $S$. Namely, for each pair there exist an 
oracle $\mathcal{O}$ and a quality measure $S$ such that $\mathcal{O}$ 
satisfies one condition under $S$ but fails the other under the same $S$.
\end{lemma}

\begin{proof}
There are three different pairs of conditions, and each pair has two 
cases. For each case we exhibit an oracle $\mathcal{O}$ and a quality 
measure $S$ such that $\mathcal{O}$ satisfies one condition under $S$ 
but not the other.

\medskip
\noindent\textbf{Pair 1: GQ vs.\ FED.}

\medskip
\noindent\textit{(GQ without FED)} Fix two distinct strings 
$r^c, r^w \in \Sigma^*$ and fix any $\mu_0 \in (0,1)$. Construct an 
oracle $\mathcal{O}$ by setting $\mathcal{D}_q = \mu_0 \delta_{r^c} + 
(1-\mu_0)\delta_{r^w}$ for all $q \in \Sigma^*$, with a control 
mechanism that makes a single draw $r_q \sim \mathcal{D}_q$ per query 
and returns the same draw for all subsequent calls on the same query 
$q$. Define a quality measure $S$ by, for all $q, r \in \Sigma^*$:
\[
S(q, r) = \begin{cases} 1 & \text{if } r = r^c, \\ 0 & \text{if } 
r = r^w, \\ \tfrac{1}{2} & \text{otherwise.} \end{cases}
\]
Then $\mathbb{E}_{r_q \sim \mathcal{D}_q}[S(q, r_q)] = \mu_0 > 0$ for 
all $q$, so $\mathcal{O}$ satisfies GQ. Since the control mechanism 
makes a single draw $r_q \sim \mathcal{D}_q$ per query and returns the 
same draw for all subsequent calls, $r_i = r_j = r_q$ for all calls 
$i, j$ on the same query $q$. Therefore $S(q, r_i)$ and $S(q, r_j)$ 
are both equal to the same random variable $X = S(q, r_q)$, where the 
randomness comes entirely from the initial draw $r_q \sim \mathcal{D}_q$. 
Hence:
\[
\mathrm{Cov}[S(q, r_i), S(q, r_j)] = \mathrm{Cov}[X, X] = 
\mathrm{Var}[X] = \mu_0(1 - \mu_0) > 0,
\]
for all $i \neq j$, at any separation $m = |j - i| - 1 \geq 0$. Since 
the covariance is constant and strictly positive for all $m$, it does 
not decay with $m$. Specifically, for any proposed constants $C > 0$ 
and $\lambda > 0$, the bound $C \cdot e^{-\lambda m} \to 0$ as 
$m \to \infty$, while $\mathrm{Cov}[S(q, r_i), S(q, r_j)] = 
\mu_0(1-\mu_0) > 0$ remains constant. Hence no such $C$ and $\lambda$ 
exist, violating FED.

\medskip
\noindent\textit{(FED without GQ)} Fix two distinct strings 
$r^c, r^w \in \Sigma^*$. Construct an oracle $\mathcal{O}$ with the 
same control mechanism as above---making a single draw $r_q \sim 
\mathcal{D}_q$ per query and returning the same draw for all subsequent 
calls on the same query---by fixing some $q_0 \in \Sigma^*$ and setting 
$\mathcal{D}_{q_0} = \delta_{r^w}$ and $\mathcal{D}_q = \delta_{r^c}$ 
for all $q \neq q_0$. Define $S$ as before: $S(q, r) = 1$ if $r = r^c$, 
$S(q, r) = 0$ if $r = r^w$, and $S(q, r) = 1/2$ otherwise. Since the 
control mechanism returns the same draw $r_q$ for all calls on the same 
query $q$, we have $r_i = r_j = r_q$ for all calls $i, j$ on the same 
query $q$. Therefore $S(q, r_i)$ and $S(q, r_j)$ are both equal to the 
same random variable $X = S(q, r_q)$, giving $\mathrm{Cov}[S(q, r_i), 
S(q, r_j)] = \mathrm{Var}[X] = 0$ for all $q \in \Sigma^*$ and all 
$i \neq j$. Hence $\mathcal{O}$ satisfies FED for any $C > 0$ and 
$\lambda > 0$. However, $\mathbb{E}_{r_{q_0} \sim 
\mathcal{D}_{q_0}}[S(q_0, r_{q_0})] = S(q_0, r^w) = 0$, so 
$\mathcal{O}$ fails GQ.

\medskip
\noindent\textbf{Pair 2: GQ vs.\ DEQ.}

\medskip
\noindent\textit{(GQ without DEQ)} Fix three distinct strings 
$r_1, r_2, r_3 \in \Sigma^*$ and construct an oracle $\mathcal{O}$ 
by setting:
\[
\begin{aligned}
\mathcal{D}_q &= \tfrac{2}{3}\delta_{r_1} + \tfrac{1}{3}\delta_{r_2} 
\quad \text{for all } q \text{ of even length,} \\
\mathcal{D}_q &= \tfrac{2}{3}\delta_{r_1} + \tfrac{1}{3}\delta_{r_3} 
\quad \text{for all } q \text{ of odd length.}
\end{aligned}
\]
Define a quality measure $S$ by, for all $q, r \in \Sigma^*$:
\[
S(q, r) = \begin{cases} 1 & \text{if } r = r_1, \\ 0 & \text{if } 
r = r_2 \text{ or } r = r_3, \\ \tfrac{1}{2} & \text{otherwise.} 
\end{cases}
\]
Then $\mathbb{E}_{r \sim \mathcal{D}_q}[S(q,r)] = \tfrac{2}{3}$ for 
every query $q$, so $\mathcal{O}$ satisfies GQ. However, for any 
even-length $q$ and odd-length $q'$, $\mathcal{D}_q \neq 
\mathcal{D}_{q'}$ but $\mathbb{E}_{r \sim \mathcal{D}_q}[S(q, r)] = 
\mathbb{E}_{r' \sim \mathcal{D}_{q'}}[S(q', r')] = \tfrac{2}{3}$, 
so $\mathcal{O}$ fails DEQ.

\medskip
\noindent\textit{(DEQ without GQ)} Enumerate all strings in $\Sigma^*$ 
as $q_0, q_1, \ldots$ and for each $q_i$ fix two distinct strings 
$r^c_i, r^w_i \in \Sigma^*$ such that all strings $r^c_0, r^w_0, 
r^c_1, r^w_1, \ldots$ are pairwise distinct; this is always possible 
since $\Sigma^*$ is infinite. Construct an oracle $\mathcal{O}$ by 
setting, for all $i \geq 0$:
\begin{align*}
\mathcal{D}_{q_i} &= \mu_i\delta_{r^c_i} + (1-\mu_i)\delta_{r^w_i}, \\
\mu_i &= i/(i+1).
\end{align*}
Note that $\mu_0 = 0$, so $\mathcal{D}_{q_0} = \delta_{r^w_0}$. Define 
a quality measure $S$ by $S(q_i, r^c_i) = 1$, $S(q_i, r^w_i) = 0$, 
and $S(q_i, r) = \tfrac{1}{2}$ for all other $r \in \Sigma^*$. The 
distributions $\mathcal{D}_{q_i}$ are all distinct since they have 
disjoint supports. Since each $\mathcal{D}_{q_i}$ is supported on 
$\{r^c_i, r^w_i\}$ only, $\mathbb{E}_{r \sim \mathcal{D}_{q_0}}
[S(q_0,r)] = 0$ and $\mathbb{E}_{r \sim \mathcal{D}_{q_i}}[S(q_i,r)] 
= i/(i+1)$ for $i \geq 1$, which are all distinct, so $\mathcal{O}$ 
satisfies DEQ. But $\mathbb{E}_{r \sim \mathcal{D}_{q_0}}[S(q_0,r)] 
= 0$, so $\mathcal{O}$ fails GQ.

\medskip
\noindent\textbf{Pair 3: FED vs.\ DEQ.}

\medskip
\noindent\textit{(FED without DEQ)} Enumerate all strings in $\Sigma^*$ 
as $q_0, q_1, \ldots$ and for each $q_i$ fix two distinct strings 
$r^c_i, r^w_i \in \Sigma^*$ such that all strings $r^c_0, r^w_0, 
r^c_1, r^w_1, \ldots$ are pairwise distinct. Fix $\mu_0 \in (0,1]$. 
Construct an oracle $\mathcal{O}$ by setting, for all $i \geq 0$:
\[
\mathcal{D}_{q_i} = \mu_0\delta_{r^c_i} + (1-\mu_0)\delta_{r^w_i},
\]
with a control mechanism that makes a fresh independent draw from 
$\mathcal{D}_{q_i}$ on each call on $q_i$, rather than reusing the 
same draw. Since the supports of $\mathcal{D}_{q_i}$ and 
$\mathcal{D}_{q_j}$ are disjoint for $i \neq j$ by pairwise 
distinctness of the strings, we have $\mathcal{D}_{q_i} \neq 
\mathcal{D}_{q_j}$ for all $i \neq j$. Define $S$ as in the DEQ 
without GQ construction: $S(q_i, r^c_i) = 1$, $S(q_i, r^w_i) = 0$, 
and $S(q_i, r) = \tfrac{1}{2}$ for all other $r \in \Sigma^*$. Since 
calls are independent, $\mathrm{Cov}_{r_a, r_b \sim 
\mathcal{D}_{q_i}}[S(q_i, r_a), S(q_i, r_b)] = 0$ for all $a \neq b$, 
satisfying FED for any $C > 0$ and $\lambda > 0$. However, 
$\mathbb{E}_{r \sim \mathcal{D}_{q_i}}[S(q_i, r)] = \mu_0$ for all 
$i \geq 0$, so all queries share the same expected quality despite 
having distinct response distributions, and $\mathcal{O}$ fails DEQ.

\medskip
\noindent\textit{(DEQ without FED)} Enumerate all strings in $\Sigma^*$ 
as $q_0, q_1, \ldots$ and for each $q_i$ fix two distinct strings 
$r^c_i, r^w_i \in \Sigma^*$ such that all strings $r^c_0, r^w_0, 
r^c_1, r^w_1, \ldots$ are pairwise distinct. Construct an oracle 
$\mathcal{O}$ by setting, for all $i \geq 0$:
\begin{align*}
\mathcal{D}_{q_i} &= \mu_i\delta_{r^c_i} + (1-\mu_i)\delta_{r^w_i}, \\
\mu_i &= 1/(i+2),
\end{align*}
with a control mechanism that makes one random draw for each $q_i$ and 
returns the same draw for every subsequent call on the same query $q_i$. 
The denominator $i+2$ rather than $i+1$ ensures $\mu_0 = \tfrac{1}{2} 
> 0$ at $i = 0$, so that GQ holds for all queries. The distributions 
$\mathcal{D}_{q_i}$ are all distinct since they have disjoint supports. 
Define $S$ as in the DEQ without GQ construction. Since each 
$\mathcal{D}_{q_i}$ is supported on $\{r^c_i, r^w_i\}$ only, 
$\mathbb{E}_{r \sim \mathcal{D}_{q_i}}[S(q_i, r)] = \mu_i = 1/(i+2)$ 
for all $i$, which are all distinct, so $\mathcal{O}$ satisfies DEQ. 
However, by the same argument as GQ without FED, $\mathrm{Cov}[S(q_i, 
r_a), S(q_i, r_b)] = \mathrm{Var}[S(q_i, r_{q_i})] = \mu_i(1-\mu_i) 
> 0$ for all $a \neq b$, where $r_{q_i} \sim \mathcal{D}_{q_i}$ is 
the initial draw on query $q_i$ and $r_a = r_b = r_{q_i}$. Since this 
covariance is constant and strictly positive for all $m$, it does not 
decay with $m$. Specifically, for any proposed constants $C > 0$ and 
$\lambda > 0$, the bound $C \cdot e^{-\lambda m} \to 0$ as $m \to 
\infty$, while $\mathrm{Cov}[S(q_i, r_a), S(q_i, r_b)] = 
\mu_i(1-\mu_i) > 0$ remains constant. Hence no such $C$ and $\lambda$ 
exist, violating FED.
\end{proof}

\begin{lemma}[BE implies GQ and FED] \label{lemma:3}
Every BE oracle satisfies GQ and FED with respect to the same quality 
measure $S$.
\end{lemma}

\begin{proof}
\textit{BE implying GQ:} Under a BE oracle, $S(q, r) \in \{0,1\}$ 
for all $q \in \Sigma^*$ and all $r \in \Sigma^*$, and 
$\mathbb{E}_{r \sim \mathcal{D}_q}[S(q,r)] \geq 1/2 + \eta > 0$ for 
all $q \in \Sigma^*$ and a fixed $\eta > 0$. Hence the oracle satisfies 
GQ with $\mu_0 = 1/2 + \eta$.

\medskip
\noindent\textit{BE implying FED:} Under a BE oracle, calls are 
independent, so:
\[
\mathrm{Cov}_{r_i, r_j \sim \mathcal{D}_q}[S(q, r_i), S(q, r_j)] = 0
\]
for all $i \neq j$ and all $q \in \Sigma^*$. Since the covariance is 
identically $0$, the FED constraint is satisfied for any $\lambda > 0$ 
and $C > 0$. Hence the oracle satisfies FED. 
\end{proof}

\begin{lemma}[BE does not imply DEQ]
There exists a BE oracle that fails DEQ under the same quality measure 
$S$.
\end{lemma}

\begin{proof}
Fix three distinct strings $r^c, r^w, r^{w'} \in \Sigma^*$ and fix 
$\eta \in (0, 1/2)$. Define $S(q, r) = 1$ if $r = r^c$ and $S(q, r) 
= 0$ otherwise. Construct an oracle $\mathcal{O}$ by setting:
\[
\mathcal{D}_q = \begin{cases} 
(1/2 + \eta)\delta_{r^c} + (1/2 - \eta)\delta_{r^w} & \text{if } 
|q| \text{ is even,} \\
(1/2 + \eta)\delta_{r^c} + (1/2 - \eta)\delta_{r^{w'}} & \text{if } 
|q| \text{ is odd,}
\end{cases}
\]
with responses drawn independently on each query. Then 
$\mathbb{E}_{r \sim \mathcal{D}_q}[S(q, r)] = 1/2 + \eta > 1/2$ for 
all $q \in \Sigma^*$, and responses are independent, so $\mathcal{O}$ 
satisfies BE. For any even-length $q$ and odd-length $q'$, 
$\mathcal{D}_q \neq \mathcal{D}_{q'}$ since $r^w \neq r^{w'}$, yet 
$\mathbb{E}_{r \sim \mathcal{D}_q}[S(q, r)] = \mathbb{E}_{r' \sim 
\mathcal{D}_{q'}}[S(q', r')] = 1/2 + \eta$. Hence $\mathcal{O}$ 
fails DEQ. 
\end{proof}

\begin{lemma}[BE does not imply $\neg$DEQ]
There exists a BE oracle that satisfies DEQ under the same quality 
measure $S$.
\end{lemma}

\begin{proof}
Fix three distinct strings $r^c, r^w_1, r^w_2 \in \Sigma^*$ and 
fix $\eta_1, \eta_2 \in (0, 1/2)$ with $\eta_1 \neq \eta_2$. 
Define $S(q, r) = 1$ if $r = r^c$ and $S(q, r) = 0$ otherwise. 
Construct an oracle $\mathcal{O}$ by setting:
\[
\mathcal{D}_q = \begin{cases}
(1/2 + \eta_1)\delta_{r^c} + (1/2 - \eta_1)\delta_{r^w_1} 
& \text{if } |q| \text{ is even,} \\
(1/2 + \eta_2)\delta_{r^c} + (1/2 - \eta_2)\delta_{r^w_2} 
& \text{if } |q| \text{ is odd,}
\end{cases}
\]  
with responses drawn independently on each query. Then 
$\mathbb{E}_{r \sim \mathcal{D}_q}[S(q, r)] \geq 1/2 + 
\min(\eta_1, \eta_2) > 1/2$ for all $q \in \Sigma^*$, and 
responses are independent, so $\mathcal{O}$ satisfies BE. For 
any even-length $q$ and odd-length $q'$, $\mathcal{D}_q \neq 
\mathcal{D}_{q'}$ and $\mathbb{E}_{r \sim \mathcal{D}_q}[S(q,r)] 
= 1/2 + \eta_1 \neq 1/2 + \eta_2 = \mathbb{E}_{r' \sim 
\mathcal{D}_{q'}}[S(q', r')]$ since $\eta_1 \neq \eta_2$. Hence 
$\mathcal{O}$ satisfies DEQ, so $\mathcal{O}$ does not satisfy 
$\neg$DEQ.
\end{proof}

\begin{remark}
BE implies neither DEQ nor its negation. An oracle satisfying BE 
may have all queries achieving the same success probability 
$1/2 + \eta$ (violating DEQ), or different queries achieving 
different success probabilities (satisfying DEQ). This gives rise 
to two distinct aspirational targets: GBE (BE + DEQ), which enables 
both probability amplification for success and adaptive quality 
optimization; and BE + $\neg$DEQ, in which every query achieves 
the same success probability regardless of formulation, enabling 
deployment robustness. These two conditions are mutually exclusive: 
an oracle cannot simultaneously satisfy DEQ and $\neg$DEQ. The 
taxonomy adopts GBE rather than BE + $\neg$DEQ as the aspirational 
target because the paper's focus is on optimizing token use, which 
requires estimating quality from observed responses; DEQ is the 
condition that makes this estimation possible.
\end{remark}

\begin{lemma}[GBE implies all conditions]
Every GBE oracle satisfies Measurable, GQ, FED, DEQ, and BE with 
respect to the same quality measure $S$.
\end{lemma}

\begin{proof}
By definition, GBE requires both BE and DEQ with respect to $S$. BE 
implies Measurable, GQ, and FED by Lemma \ref{lemma:3}. DEQ is required directly 
by the GBE definition. Hence a GBE oracle satisfies all six conditions. 

\end{proof}

\subsection{LLMs in the Taxonomy}\label{LLMs in the Taxonomy}

We assess current LLMs against each condition in turn.

\begin{itemize}
  \item \textbf{Measurable:} All LLMs satisfy Measurable: their 
  responses are drawn from well-defined distributions, making expected 
  quality and token cost well-defined. This is a necessary baseline, 
  not a useful guarantee.

  \item \textbf{Guaranteed Quality:} Current evidence suggests that 
  LLMs satisfy GQ for well-defined tasks: LLMs achieve expected quality 
  $\mu_0 > 0$ when evaluation criteria are clear. However, the true 
  per-query quality varies widely, making the uniform bound $\mu_0$ 
  loose and unknown without empirical measurement.

  \item \textbf{Fast Error Decay:} LLMs satisfy FED when called 
  statelessly: error correlations across repeated calls decay over 
  time, consistent with a finite but unknown $t_\text{corr}$. In the 
  chatbot setting, conversational history is maintained within the 
  context window, improving response relevance but allowing errors to 
  persist and compound across turns with no guaranteed decay.

  \item \textbf{Distinct Expected Quality:} Current evidence suggests 
  that LLMs satisfy DEQ: expected quality levels vary empirically 
  across queries, enabling task-specific calibration. However, 
  provider-side updates to the model or deployment can shift quality 
  levels discontinuously, making prior calibration unreliable after 
  an update.

  \item \textbf{Bounded Error (BE):} LLMs fail BE. This failure is 
  not an engineering limitation but a consequence of what LLMs 
  fundamentally are: stochastic, open-ended systems without a fixed 
  correctness criterion.

  \item \textbf{Global Bounded Error (GBE):} Since LLMs fail BE, they 
  also fail GBE. GBE is the aspirational target for LLM technology.
\end{itemize}

This classification tells us exactly which structural properties of 
the oracle framework apply to LLMs and which do not. All theorems on 
token complexity in this paper require at most GQ, and therefore apply 
to LLM oracles. GQ is the weakest condition needed: it ensures that expected quality 
under the given quality measure $S$ is bounded away from zero on every 
query, which is sufficient to define a non-trivial quality constraint 
and guarantee the feasibility of token complexity at some positive 
threshold.

\section{Token Complexity}\label{Token Complexity}

We define token complexity in four steps: the tokenizer and token 
count, the weighted token cost function, token complexity itself, and 
the complexity frontier characterizing tradeoffs over tokens, time, 
and space. We close with a discussion situating token complexity 
relative to Kolmogorov complexity and rate-distortion theory.

\textbf{Tokenizer.} Token complexity is measured via a tokenizer 
$\tau: \Sigma^* \rightarrow \mathcal{T}^*$, a function mapping strings 
to token sequences, with an associated deterministic decoding function 
$\rho: \mathcal{T}^* \rightarrow \Sigma^*$ satisfying $\rho(\tau(s)) 
= s$ for all $s \in \Sigma^*$. The token vocabulary $\mathcal{T}$ is 
the set of individual tokens appearing in the range of $\tau$. For 
notational convenience, we use $\tau^{-1}$ to denote $\rho$ throughout 
the paper. This definition reflects real tokenization practice, where 
schemes such as Byte Pair Encoding (BPE)~\cite{Sennrich2016}, 
WordPiece~\cite{Schuster2012}, or SentencePiece~\cite{Kudo2018} map 
raw text to subword token sequences.

\subsection{Token Count}\label{Token Count}

Let $\mathcal{M} = (M, \mathcal{O})$ be an AOTM for solving task 
$T = (X, Y, S, \mathcal{D}_X)$.

\begin{definition}[Token Count]
The \textbf{query token count}, \textbf{response token count}, and 
\textbf{total token count} of $\mathcal{M}$ on task instance $x$ are 
defined as:
\begin{align*}
\mathrm{tok}_{\mathcal{M},Q}(x) &= \sum_i |\tau(q_i)|, \\
\mathrm{tok}_{\mathcal{M},R}(x) &= \sum_i |\tau(r_i)|, \\
\mathrm{tok}_{\mathcal{M}}(x) &= \mathrm{tok}_{\mathcal{M},Q}(x) + 
\mathrm{tok}_{\mathcal{M},R}(x),
\end{align*}
where $|\tau(q_i)|$ and $|\tau(r_i)|$ denote the number of tokens in 
the query and response of the $i$-th oracle call. No other quantities 
are included; in particular, overhead terms, session structure, and 
streaming behavior are accounted for by $\mathrm{TIME}_\mathcal{M}$ 
and $\mathrm{SPACE}_\mathcal{M}$.
\end{definition}

\begin{definition}[Expected Token Count]
The \textbf{expected query token count}, \textbf{expected response 
token count}, and \textbf{expected total token count} of $\mathcal{M}$ 
are defined as:
\begin{align*}
\mathrm{tok}_{\mathcal{M},Q} &= \mathbb{E}_{x \sim \mathcal{D}_X}
[\mathrm{tok}_{\mathcal{M},Q}(x)], \\
\mathrm{tok}_{\mathcal{M},R} &= \mathbb{E}_{x \sim \mathcal{D}_X}
[\mathrm{tok}_{\mathcal{M},R}(x)], \\
\mathrm{tok}_{\mathcal{M}} &= \mathrm{tok}_{\mathcal{M},Q} + 
\mathrm{tok}_{\mathcal{M},R}.
\end{align*}
\end{definition}

\subsection{Token Cost}\label{Token Cost}

Let $\alpha, \beta > 0$ denote the per-token costs for querying and 
receiving responses from $\mathcal{O}$. The token cost of $\mathcal{M}$ on task instance $x$ is:
\[
\mathrm{TOK}_{\mathcal{M}}(x; \alpha, \beta) = \alpha \cdot 
\mathrm{tok}_{\mathcal{M},Q}(x) + \beta \cdot 
\mathrm{tok}_{\mathcal{M},R}(x),
\]
and the task-level token cost is:
\[
\mathrm{TOK}_{\mathcal{M}}(\alpha, \beta) = \alpha \cdot 
\mathrm{tok}_{\mathcal{M},Q} + \beta \cdot 
\mathrm{tok}_{\mathcal{M},R}.
\]
We write $\mathrm{TOK}_{\mathcal{M}}$ for $\mathrm{TOK}_{\mathcal{M}}
(\alpha, \beta)$ when $\alpha$ and $\beta$ are clear from context.

Throughout this paper, $\alpha$ and $\beta$ represent per-token 
cost parameters that can be interpreted either as assigned prices 
(as in commercial LLM APIs, where query and response tokens are 
priced differently) or as intrinsic costs (as in open-source LLMs, 
where the same asymmetry arises through compute). The terms 
``price'' and ``cost'' are used interchangeably to refer to 
$\alpha$ and $\beta$.

\subsection{Token Complexity}\label{Token Complexity-2}

\begin{definition}[Token Complexity]
The \textbf{token complexity} of task $T$ at quality threshold 
$\theta \in (0, 1]$ is:
\[
\kappa_T(\theta; \alpha, \beta) = \min_{\mathcal{M} \text{ for } T} 
\mathrm{TOK}_{\mathcal{M}}, \quad \text{subject to } \mathbb{E}_{x 
\sim \mathcal{D}_X}[S(x, \mathcal{M}(x))] \geq \theta,
\]
where the minimum is over all AOTMs $\mathcal{M} = (M, \mathcal{O})$ 
for task $T$.
\end{definition}

\begin{remark}
The minimum ranges over all querying strategies for solving the 
underlying task, including non-adaptive and adaptive strategies. Token 
complexity thus captures the minimum cost of the best possible strategy 
for solving the task.
\end{remark}

We write $\kappa_T(\theta; \alpha, \beta)$ to make the dependence on 
$T$ explicit. When the task is clear from context, we omit the 
subscript and write $\kappa(\theta; \alpha, \beta)$, or simply 
$\kappa(\theta)$ when $\alpha$ and $\beta$ are also clear from context.

The randomness of $\mathcal{M}(x)$ comes from two sources: $M$'s 
independent private random source and $\mathcal{O}$'s stochastic 
response distributions $\{\mathcal{D}_q\}$. The expectation 
$\mathbb{E}_{x \sim \mathcal{D}_X}[S(x, \mathcal{M}(x))]$ is taken 
jointly over both sources and over $x \sim \mathcal{D}_X$.

When no AOTM satisfies the quality constraint at threshold $\theta$, 
the feasible set is empty. By convention, the minimum over an empty 
set is $+\infty$, so $\kappa_T(\theta; \alpha, \beta) = \infty$ in 
this case.

\begin{remark}
GQ oracles guarantee that token complexity is well-defined for $\theta 
\in (0, \mu_0]$: the feasible set is non-empty since $\mathbb{E}_{x 
\sim \mathcal{D}_X}[S(x, \mathcal{M}(x))] \geq \mu_0 > 0$ for some 
AOTM $\mathcal{M}$, giving $\kappa_T(\theta; \alpha, \beta) < \infty$. 
Without GQ, the feasible set may be empty at every positive threshold, 
making token complexity trivially infinite.
\end{remark}

\begin{remark}
For a positive integer $n$, the \textbf{token complexity class} 
$\mathrm{TC}(n; \theta, \alpha, \beta)$ is the set of all tasks $T$ 
with $\kappa_T(\theta; \alpha, \beta) \leq n$:
\[
\mathrm{TC}(n; \theta, \alpha, \beta) = \{T : \kappa_T(\theta; \alpha, 
\beta) \leq n\}.
\]
This is analogous to classical complexity classes defined by resource 
bounds, such as This is analogous to classical complexity classes defined by resource 
bounds, such as P or PSPACE. A task is in 
$\mathrm{TC}(n; \theta, \alpha, \beta)$ if it can be solved within $n$ 
tokens at quality $\theta$. The exact-value class $\{T : 
\kappa_T(\theta; \alpha, \beta) = n\}$ is the difference 
$\mathrm{TC}(n; \theta, \alpha, \beta) \setminus \mathrm{TC}(n-1; 
\theta, \alpha, \beta)$. Unlike classical complexity classes, which 
classify problems by asymptotic complexity growth rate, token 
complexity classes are exact; they distinguish tasks that differ by 
even a single token. The structure of these classes---how they nest, 
how they change with $(\theta, \alpha, \beta)$, and which tasks fall 
in the same class---is beyond the scope of this paper but is a natural 
direction for future work.
\end{remark}

\subsection{Complexity Frontier}\label{Complexity Frontier}

\begin{definition}[Complexity Frontier]
The \textbf{complexity frontier} of task $T$ at quality threshold 
$\theta \in (0, 1]$ over a subset of instances $X_n \subseteq X$ is:
\[
F_T(\theta; X_n) = \left\{(k, t, s) \in \mathbb{R}^3 : \exists\, 
\mathcal{M} = (M, \mathcal{O}) \text{ s.t. }
\begin{aligned}
&\mathbb{E}_{x \sim \mathcal{D}_X|_{X_n}}[\mathrm{tok}_{\mathcal{M}}(x)] 
\leq k, \\
&\mathbb{E}_{x \sim \mathcal{D}_X|_{X_n}}[\mathrm{TIME}_{\mathcal{M}}(x)] 
\leq t, \\
&\max_{x \in X_n} \mathrm{SPACE}_{\mathcal{M}}(x) \leq s, \\
&\mathbb{E}_{x \sim \mathcal{D}_X|_{X_n}}[S(x, \mathcal{M}(x))] \geq 
\theta
\end{aligned}
\right\},
\]
where $\mathcal{D}_X|_{X_n}$ denotes $\mathcal{D}_X$ conditioned on 
$X_n$. Time and token resources are measured in expectation; space is 
measured worst-case since it must be available simultaneously for all 
computations. The full frontier $F_T(\theta) = F_T(\theta; X)$ is the 
case $X_n = X$.
\end{definition}

The complexity frontier is a subset of $\mathbb{R}^3$ characterizing 
the achievable combinations of token count, time, and space for a task, 
whose lower boundary captures all efficient operating points---resource 
bounds where no bound can be reduced while maintaining feasibility. 
Unlike the coarse-grained hierarchy of classical complexity classes, 
the frontier captures the continuous structure of resource tradeoffs, 
making explicit how token count, time, and space can be substituted 
for one another. As $\theta$ varies continuously over $[0,1]$, the 
frontier $F_T(\theta; X_n)$ traces a continuous family of feasible 
sets, in contrast to the discrete membership boundaries of classical 
complexity classes.

\begin{remark}
Restricting to finite subsets $X_n \subseteq X$ reflects reality: any 
deployed application encounters only a finite collection of task 
instances within any given time interval. The full input space $X$ is 
a theoretical abstraction that bounds what the application may 
eventually face, but actual token complexity is always computed over a 
finite sample. The frontier $F_T(\theta; X_n)$ is therefore the 
practically relevant object, and its properties---non-emptiness, 
upward-closure, and convexity---hold for any finite $X_n$, with proofs 
that require no additional assumptions beyond finiteness. The time 
interval determines the finiteness of $X_n$: it is an upper time bound 
on the computation of all but the last task instance, whose computation 
may not have completed before the interval ends.
\end{remark}

\begin{remark}
Using token cost $\mathrm{TOK}_{\mathcal{M}}(x) = \alpha \cdot 
\mathrm{tok}_{\mathcal{M},Q}(x) + \beta \cdot 
\mathrm{tok}_{\mathcal{M},R}(x)$ as the first coordinate gives a 
price-dependent frontier $F_T(\theta; X_n; \alpha, \beta)$. This has 
the advantage of directly encoding the monetary cost at given prices, 
making the connection to $\kappa_T(\theta; \alpha, \beta)$ immediate. 
The disadvantage is that the frontier changes with prices, obscuring 
the underlying computational structure: two systems with identical 
computational behavior but different price environments would have 
different frontiers. The price-independent formulation adopted here 
separates the geometric structure of the frontier---what is 
computationally achievable---from the economic question of which 
achievable point is cheapest at given prices.
\end{remark}

\subsection{Relationship to Kolmogorov Complexity and Rate-Distortion Theory}\label{Relationship to Kolmogorov Complexity and Rate-Distortion Theory}

Token complexity is structurally related to two classical theories; 
making these connections explicit clarifies what is new.

\textbf{Kolmogorov complexity}~\cite{Kolmogorov1965, Li2008} defines 
the complexity of a string $x$, denoted $K(x)$, as the length of the 
shortest program that produces it: a minimum over all possible 
descriptions. Token complexity is analogously defined as a minimum 
over all possible machines. Both measure the irreducible cost of a 
computational task as a minimum over all possible descriptions or 
strategies. Three differences are worth stating precisely:

\begin{itemize}
  \item \textbf{Object of study:} Kolmogorov complexity is a property 
  of a single string. Token complexity is a property of a task defined 
  by a distribution $\mathcal{D}_X$: complexity is a weighted average 
  over task instances, with weights given by $\mathcal{D}_X$.

  \item \textbf{Stochasticity:} Kolmogorov complexity is deterministic: 
  the shortest program either produces $x$ or it does not. Token 
  complexity is stochastic: quality is an expectation over oracle 
  responses and the input distribution, with the threshold $\theta$ 
  parameterizing the required quality level.

  \item \textbf{Computability:} Kolmogorov complexity is uncomputable: 
  no algorithm can compute $K(x)$ for all $x$. Whether token complexity 
  $\kappa(\theta; \alpha, \beta)$ is computable remains an open problem 
  (Open Problem O1).
\end{itemize}

The closer formal analog is \textbf{rate-distortion 
theory}~\cite{Shannon1959, Cover2006}, which asks for the minimum 
communication cost to reconstruct a source within a given distortion 
level. Token complexity shares this structure: minimum cost to achieve 
a given quality of output, as a function of how much error is 
tolerable. The structural parallels suggest that techniques from both 
Kolmogorov complexity and rate-distortion theory may prove useful in 
studying token complexity.

\section{Properties of Token Complexity}
\label{Properties of Token Complexity}

In practice, achieving higher quality tends to require more targeted 
token use: more precise queries, selective retries, or output 
verification. The following two theorems formalize this relationship. 
Both concern $\kappa$ as a function of $\theta$ alone, holding $T$, 
$\alpha$, and $\beta$ fixed; we write $\kappa(\theta)$ for 
$\kappa_T(\theta; \alpha, \beta)$ to simplify notation.

\begin{theorem}[Monotonicity]
$\kappa(\theta)$ is monotone non-decreasing in $\theta$: if $\theta_2 
\geq \theta_1$, then $\kappa(\theta_2) \geq \kappa(\theta_1)$.
\end{theorem}

\begin{proof}
If $\theta_2 \geq \theta_1$, any AOTM satisfying the quality constraint 
at $\theta_2$ also satisfies it at $\theta_1$. The feasible set at 
$\theta_1$ is therefore no smaller, so the minimum is no greater. 

\end{proof}

\begin{theorem}[Convexity]
$\kappa(\theta)$ is convex in $\theta$: for all $\theta_1, \theta_2 
\in (0,1]$ and all $\gamma \in [0,1]$:
\[
\kappa(\gamma \cdot \theta_1 + (1-\gamma) \cdot \theta_2) \leq 
\gamma \cdot \kappa(\theta_1) + (1-\gamma) \cdot \kappa(\theta_2).
\]
\end{theorem}

\begin{proof}
Let $\mathcal{M}_1$ and $\mathcal{M}_2$ be AOTMs achieving quality 
$\theta_1$ and $\theta_2$ with token costs $c_1 = 
\mathrm{TOK}_{\mathcal{M}_1}$ and $c_2 = \mathrm{TOK}_{\mathcal{M}_2}$ 
respectively. Construct $\mathcal{M}_\gamma$ as follows: as a 
preprocessing step, $\mathcal{M}_\gamma$ flips a biased coin with 
heads probability $\gamma$ using its random tape, independently of 
the input $x$. It then runs $\mathcal{M}_1$ on heads and $\mathcal{M}_2$ 
on tails. By linearity of expectation:
\begin{align*}
\mathbb{E}_{x \sim \mathcal{D}_X}[S(x, \mathcal{M}_\gamma(x))] 
&= \gamma \cdot \mathbb{E}_{x \sim \mathcal{D}_X}[S(x, \mathcal{M}_1(x))] 
+ (1-\gamma) \cdot \mathbb{E}_{x \sim \mathcal{D}_X}[S(x, 
\mathcal{M}_2(x))] \\
&\geq \gamma \cdot \theta_1 + (1-\gamma) \cdot \theta_2, \\
\mathrm{TOK}_{\mathcal{M}_\gamma} 
&= \gamma \cdot \mathrm{TOK}_{\mathcal{M}_1} + (1-\gamma) \cdot 
\mathrm{TOK}_{\mathcal{M}_2} \\
&= \gamma \cdot c_1 + (1-\gamma) \cdot c_2.
\end{align*}
So $\mathcal{M}_\gamma$ is feasible at threshold $\gamma \cdot \theta_1 
+ (1-\gamma) \cdot \theta_2$, giving:
\[
\kappa(\gamma \cdot \theta_1 + (1-\gamma) \cdot \theta_2) \leq 
\gamma \cdot c_1 + (1-\gamma) \cdot c_2.
\]
Since this holds for any $\mathcal{M}_1$ achieving $\theta_1$ and any 
$\mathcal{M}_2$ achieving $\theta_2$, taking the minimum over all such 
pairs gives the result. 
\end{proof}

\begin{corollary}[Increasing Marginal Cost]
For any $\theta_1 < \theta_2 < \theta_3$:
\[
\frac{\kappa(\theta_2) - \kappa(\theta_1)}{\theta_2 - \theta_1} \leq 
\frac{\kappa(\theta_3) - \kappa(\theta_2)}{\theta_3 - \theta_2}.
\]
\end{corollary}

\begin{proof}[Proof Sketch]
Write $\theta_2 = \gamma\theta_1 + (1-\gamma)\theta_3$ where 
$\gamma = (\theta_3 - \theta_2)/(\theta_3 - \theta_1) \in (0,1)$. 
The inequality follows directly from convexity of $\kappa(\theta)$. 

\end{proof}

\begin{remark}
The corollary says that the token cost per unit of quality is 
non-decreasing: moving from $\theta_1$ to $\theta_2$ costs fewer 
tokens per unit of quality than moving from $\theta_2$ to $\theta_3$. 
Quality improvements become progressively more expensive as quality 
approaches the ceiling $\bar{\theta}$. Intuitively, the first 
improvement is cheap---a slightly better query suffices---but as 
quality approaches $\bar{\theta}$, each additional gain requires 
progressively more tokens: more careful query formulation, more oracle 
calls, more output verification, until the cost becomes infinite at 
the ceiling.
\end{remark}

\begin{theorem}[Price Sensitivity]\label{thm:price_sensitivity}
Let $T = (X, Y, S, \mathcal{D}_X)$ be a task. For any $\alpha, \beta, 
\alpha', \beta' > 0$ and any $\theta \in (0,1]$, let $\mathcal{M}_1$ 
and $\mathcal{M}_2$ be feasible AOTMs for $\kappa_T(\theta; \alpha, 
\beta)$ and $\kappa_T(\theta; \alpha', \beta')$ respectively. Then:
\begin{align*}
|\kappa_T(\theta; \alpha, \beta) - \kappa_T(\theta; \alpha', \beta')| 
&\leq \max(\mathrm{tok}_{\mathcal{M}_1,Q}, \mathrm{tok}_{\mathcal{M}_2,Q}) 
\cdot |\alpha - \alpha'| \\
&\quad + \max(\mathrm{tok}_{\mathcal{M}_1,R}, 
\mathrm{tok}_{\mathcal{M}_2,R}) \cdot |\beta - \beta'|.
\end{align*}
\end{theorem}

\begin{proof}
Since the quality constraint $\mathbb{E}_{x \sim \mathcal{D}_X}
[S(x, \mathcal{M}_1(x))] \geq \theta$ depends only on $\mathcal{M}_1$ 
and $S$, not on prices, $\mathcal{M}_1$ satisfies the quality constraint 
at $\theta$ under $(\alpha', \beta')$. Similarly, $\mathcal{M}_2$ 
satisfies the quality constraint at $\theta$ under $(\alpha, \beta)$. 
Each is therefore a feasible solution: $\mathcal{M}_2$ for 
$\kappa_T(\theta; \alpha, \beta)$ and $\mathcal{M}_1$ for 
$\kappa_T(\theta; \alpha', \beta')$.

Therefore:
\[
\kappa_T(\theta; \alpha, \beta) \leq \alpha \cdot 
\mathrm{tok}_{\mathcal{M}_2,Q} + \beta \cdot \mathrm{tok}_{\mathcal{M}_2,R}.
\]
Since $\mathcal{M}_2$ is also feasible under $(\alpha', \beta')$:
\[
\kappa_T(\theta; \alpha', \beta') \leq \alpha' \cdot 
\mathrm{tok}_{\mathcal{M}_2,Q} + \beta' \cdot \mathrm{tok}_{\mathcal{M}_2,R}.
\]
Subtracting:
\begin{align*}
\kappa_T(\theta; \alpha, \beta) - \kappa_T(\theta; \alpha', \beta')
&\leq (\alpha - \alpha') \cdot \mathrm{tok}_{\mathcal{M}_2,Q} + 
(\beta - \beta') \cdot \mathrm{tok}_{\mathcal{M}_2,R} \\
&\leq \mathrm{tok}_{\mathcal{M}_2,Q} \cdot |\alpha - \alpha'| + 
\mathrm{tok}_{\mathcal{M}_2,R} \cdot |\beta - \beta'|.
\end{align*}
The symmetric argument gives:
\[
\kappa_T(\theta; \alpha', \beta') - \kappa_T(\theta; \alpha, \beta) 
\leq \mathrm{tok}_{\mathcal{M}_1,Q} \cdot |\alpha - \alpha'| + 
\mathrm{tok}_{\mathcal{M}_1,R} \cdot |\beta - \beta'|.
\]
Combining the two one-sided bounds gives the absolute value 
inequality. 
\end{proof}

\begin{remark}
Price-sensitivity shows that
small price changes produce bounded changes in token complexity, making 
the effect of API price changes on optimal cost predictable and 
amenable to sensitivity analysis.
\end{remark}

\begin{definition}[Price-Specific Dominance Order]
\label{def:price_specific_dominance}
Let $T_1$ and $T_2$ be tasks and let $\alpha, \beta > 0$. We say 
$T_1 \preceq_\kappa^{(\alpha,\beta)} T_2$ if 
$\kappa_{T_1}(\theta; \alpha, \beta) \leq \kappa_{T_2}(\theta; 
\alpha, \beta)$ for all $\theta \in (0,1]$.
\end{definition}

\begin{theorem}[Price-Relativity of Task Ordering]
\label{thm:price_relativity}
There exist tasks $T_1$, $T_2$ and a threshold $K > 0$ such that 
$T_1 \preceq_\kappa^{(\alpha, \beta)} T_2$ when $\alpha / \beta > K$, 
$T_2 \preceq_\kappa^{(\alpha, \beta)} T_1$ when $\alpha / \beta < K$, 
and $\kappa_{T_1}(\theta; \alpha, \beta) = \kappa_{T_2}(\theta; 
\alpha, \beta)$ for all $\theta$ when $\alpha / \beta = K$.
\end{theorem}

\begin{proof}
Choose tasks $T_1$ and $T_2$ whose optimal AOTMs $\mathcal{M}_1$ 
and $\mathcal{M}_2$ at the same quality threshold $\theta$ satisfy 
$\mathrm{tok}_{\mathcal{M}_1, Q} > \mathrm{tok}_{\mathcal{M}_2, Q}$ 
and $\mathrm{tok}_{\mathcal{M}_2, R} > \mathrm{tok}_{\mathcal{M}_1, R}$: 
$T_1$ uses more query tokens and $T_2$ uses more response tokens. 
For any prices $(\alpha, \beta)$, the token complexities are:
\begin{align*}
\kappa_{T_1}(\theta; \alpha, \beta) &= \alpha \cdot 
\mathrm{tok}_{\mathcal{M}_1, Q} + \beta \cdot 
\mathrm{tok}_{\mathcal{M}_1, R}, \\
\kappa_{T_2}(\theta; \alpha, \beta) &= \alpha \cdot 
\mathrm{tok}_{\mathcal{M}_2, Q} + \beta \cdot 
\mathrm{tok}_{\mathcal{M}_2, R}.
\end{align*}
The difference is:
\begin{align}
\kappa_{T_1}(\theta; \alpha, \beta) - \kappa_{T_2}(\theta; \alpha, 
\beta) = \alpha(\mathrm{tok}_{\mathcal{M}_1, Q} - 
\mathrm{tok}_{\mathcal{M}_2, Q}) + \beta(\mathrm{tok}_{\mathcal{M}_1, 
R} - \mathrm{tok}_{\mathcal{M}_2, R}).\label{eq:difference}
\end{align}
Dividing by $\beta > 0$ and setting the difference to zero, 
solving for $\alpha / \beta$ gives the threshold:
\[
K = \frac{\mathrm{tok}_{\mathcal{M}_2, R} - 
\mathrm{tok}_{\mathcal{M}_1, R}}{\mathrm{tok}_{\mathcal{M}_1, Q} - 
\mathrm{tok}_{\mathcal{M}_2, Q}}.
\]
Since both numerator and denominator are strictly positive by 
assumption, we have $K > 0$. When $\alpha / \beta < K$ the 
difference~(\ref{eq:difference}) is negative, giving 
$T_1 \preceq_\kappa^{(\alpha, \beta)} T_2$; when $\alpha / \beta 
> K$ it is positive, giving $T_2 \preceq_\kappa^{(\alpha, \beta)} 
T_1$; and when $\alpha / \beta = K$ the two complexities are equal.
\end{proof}

\begin{remark}
Price-relativity shows that the relative token cost of two tasks is 
not an intrinsic property of the tasks but depends on the price ratio 
$\alpha / \beta$. This contrasts with classical complexity, where the 
hardness ordering between problems is intrinsic and independent of 
resource pricing.
\end{remark}

\section{Properties of the Complexity Frontier}
\label{Properties of the Complexity Frontier}

The complexity frontier $F_T(\theta; X_n)$ is the set of all resource 
bounds $(k, t, s)$ under which task $T$ can be solved at expected quality 
at least $\theta$ over instance subset $X_n \subseteq X$. The full 
frontier $F_T(\theta) = F_T(\theta; X)$ is the case $X_n = X$. The 
boundary $\partial F_T(\theta; X_n)$---the set of resource bounds where 
no coordinate can be reduced while maintaining feasibility---is the 
tradeoff surface. Fixing any one resource bound and projecting onto the 
remaining two yields the pairwise tradeoff curves. This section 
establishes three structural properties: (1) $F_T(\theta; X_n)$ is 
non-empty for any finite $X_n$ whenever $\theta$ is achievable, (2) the 
frontier is upward-closed, meaning any resource bound dominating a 
feasible bound is also feasible, and (3) the tradeoff surface is convex, 
from which the convexity of all pairwise tradeoff curves follows.

\begin{remark}
The complexity frontier is a geometric object with no direct classical 
analog. Classical complexity theory studies resources one at a time, 
classifying problems by asymptotic bounds on a single resource; 
time-space tradeoff results~\cite{HopcroftPaulValiant1977} establish 
relationships between time and space as asymptotic bounds, not as a 
joint feasible region. The complexity frontier departs from this in 
three ways: it is exact rather than asymptotic, it is defined jointly 
over three resources, and it is parameterized by a continuous quality 
threshold $\theta$. The tradeoff surface $\partial F_T(\theta; X_n)$ 
is the natural geometric object capturing how token cost, time, and 
space can be substituted for one another at a given quality level---an object that has no counterpart in classical complexity theory precisely because classical complexity has no stochastic oracle, no token cost, and no quality parameter.
\end{remark}

\begin{theorem}[Non-Emptiness]\label{thm:frontier_nonempty}
For any finite $X_n \subseteq X$ and any $\theta \in (0, 1]$ such 
that there exists an AOTM achieving expected quality at least $\theta$ 
on $X_n$, the complexity frontier $F_T(\theta; X_n)$ is non-empty.
\end{theorem}

\begin{proof}
Let $\mathcal{M}$ be an AOTM achieving expected quality at least 
$\theta$ on $X_n$. Since $X_n$ is finite, $\mathcal{M}$ terminates 
on every $x \in X_n$ in a finite number of transitions, uses a finite 
amount of tape, and issues a finite number of oracle calls each 
exchanging a finite number of tokens. TTaking expectations and maximum over the finite set $X_n$ gives 
finite values:
\begin{align*}
k &= \mathbb{E}_{x \sim \mathcal{D}_X|_{X_n}}
[\mathrm{tok}_{\mathcal{M}}(x)] < \infty, \\
t &= \mathbb{E}_{x \sim \mathcal{D}_X|_{X_n}}
[\mathrm{TIME}_{\mathcal{M}}(x)] < \infty, \\
s &= \max_{x \in X_n} \mathrm{SPACE}_{\mathcal{M}}(x) < \infty,
\end{align*}
so $(k, t, s) \in F_T(\theta; X_n)$. 
\end{proof}

\begin{theorem}[Upward-Closure]\label{thm:frontier_upward}
For any $X_n \subseteq X$, quality threshold $\theta$, and resource 
bound $(k, t, s) \in F_T(\theta; X_n)$, if $(k', t', s') \geq 
(k, t, s)$ componentwise, then $(k', t', s') \in F_T(\theta; X_n)$.
\end{theorem}

\begin{proof}
Let $\mathcal{M}$ be an AOTM witnessing $(k, t, s) \in F_T(\theta; 
X_n)$. The quality constraint
\[
\mathbb{E}_{x \sim \mathcal{D}_X|_{X_n}}[S(x, \mathcal{M}(x))] 
\geq \theta
\]
depends only on $\mathcal{M}$ and $S$, not on the resource bounds, 
so $\mathcal{M}$ still achieves quality $\theta$ under $(k', t', s')$. 
Since $\mathcal{M}$'s resource usage satisfies 
$\mathrm{tok}_{\mathcal{M}} \leq k \leq k'$, 
$\mathrm{TIME}_{\mathcal{M}} \leq t \leq t'$, and 
$\mathrm{SPACE}_{\mathcal{M}} \leq s \leq s'$, $\mathcal{M}$ 
witnesses $(k', t', s') \in F_T(\theta; X_n)$. 
\end{proof}

\begin{remark}
Upward-closure means $F_T(\theta; X_n)$ is entirely determined by its 
lower boundary. Any resource bound on or above the lower boundary is 
feasible; any bound strictly below it is not. This justifies studying 
the lower boundary as the essential object: it is the minimal description 
of what is achievable at quality $\theta$.
\end{remark}

\begin{definition}[Tradeoff Surface]\label{def:tradeoff_surface}
The \textbf{tradeoff surface} of $F_T(\theta; X_n)$ is the lower 
boundary of $F_T(\theta; X_n)$, consisting of resource bounds where no 
coordinate can be strictly reduced while maintaining feasibility:
\[
\partial F_T(\theta; X_n) = \left\{(k, t, s) \in F_T(\theta; X_n) :
\begin{aligned}
&\forall\, (k', t', s') \in F_T(\theta; X_n), \\
&k' \leq k,\, t' \leq t,\, s' \leq s \implies (k', t', s') = (k, t, s)
\end{aligned}
\right\}.
\]
A pairwise \textbf{tradeoff curve} is a cross-section of 
$\partial F_T(\theta; X_n)$ obtained by fixing one resource bound at 
a constant value.
\end{definition}

\begin{theorem}[Convexity of Tradeoff Surface]\label{thm:frontier_convex}
For any $X_n \subseteq X$ and quality threshold $\theta$, the complexity 
frontier $F_T(\theta; X_n)$ is a convex set in $\mathbb{R}^3$. 
Consequently, $\partial F_T(\theta; X_n)$ is convex, and the TOK-TIME, 
TOK-SPACE, and TIME-SPACE tradeoff curves are all convex.
\end{theorem}

\begin{proof}
Suppose $(k_1, t_1, s_1)$ and $(k_2, t_2, s_2)$ are both in 
$F_T(\theta; X_n)$, witnessed by AOTMs $\mathcal{M}_1$ and $\mathcal{M}_2$ 
respectively. For any $\gamma \in [0,1]$, construct $\mathcal{M}_\gamma$ 
as follows: as a preprocessing step, $\mathcal{M}_\gamma$ flips a biased 
coin with heads probability $\gamma$ using its random tape, independently 
of $x$. It then runs $\mathcal{M}_1$ on heads and $\mathcal{M}_2$ on 
tails. Since both achieve quality at least $\theta$, by linearity of 
expectation so does $\mathcal{M}_\gamma$. Since the coin flip is 
independent of $x$:
\begin{align*}
\mathbb{E}_{x \sim \mathcal{D}_X|_{X_n}}[\mathrm{tok}_{\mathcal{M}_\gamma}(x)]
  &= \gamma k_1 + (1-\gamma)k_2, \\
\mathbb{E}_{x \sim \mathcal{D}_X|_{X_n}}[\mathrm{TIME}_{\mathcal{M}_\gamma}(x)]
  &= \gamma t_1 + (1-\gamma)t_2.
\end{align*}
For space, $\mathcal{M}_\gamma$ uses at most $\max(s_1, s_2)$ on any 
instance. Since $\max(s_1, s_2) \geq \gamma s_1 + (1-\gamma)s_2$, the 
point $(\gamma k_1 + (1-\gamma)k_2,\ \gamma t_1 + (1-\gamma)t_2,\ 
\gamma s_1 + (1-\gamma)s_2)$ is dominated by $(\gamma k_1 + 
(1-\gamma)k_2,\ \gamma t_1 + (1-\gamma)t_2,\ \max(s_1, s_2))$, which 
is in $F_T(\theta; X_n)$, and upward-closure gives the result. Hence 
$F_T(\theta; X_n)$ is convex, its lower boundary $\partial F_T(\theta; 
X_n)$ is a convex surface, and any cross-section fixing one resource 
bound is a convex curve. 
\end{proof}

\begin{corollary}[Convexity of Tradeoff Curves]
Each pairwise tradeoff curve of $F_T(\theta; X_n)$ is convex.
\end{corollary}

\begin{proof}
Fix one resource bound at a constant value, say $s = s_0$. The 
cross-section $\{(k, t) : (k, t, s_0) \in F_T(\theta; X_n)\}$ is a 
convex set in $\mathbb{R}^2$, since the intersection of a convex set 
with a hyperplane is convex. The TOK-TIME tradeoff curve is the lower 
boundary of this convex set, and is therefore convex: moving along it 
from high token cost and low time to low token cost and high time, each 
unit reduction in token cost requires a progressively larger increase in 
time to maintain feasibility at quality $\theta$. The same argument holds 
for TOK-SPACE and TIME-SPACE. $\square$
\end{proof}

\begin{remark}
In practice, tradeoffs are not always smooth: switching between model 
tiers, caching strategies, or hardware configurations introduces discrete 
jumps in token cost, time latency, and space requirements. The convexity 
result is a theoretical idealization that approximates this discrete 
reality as a continuous surface.
\end{remark}

\section{Major Open Problems}\label{Open Problems}

The token complexity framework developed in this paper is a first step 
toward a comprehensive theory of oracle complexity. In developing it, 
we have encountered many open questions; we identify five as central, 
together defining the research agenda for token complexity theory: the 
computability of the quality ceiling, the existence of token lower 
bounds, the universality of frontier shapes across parameterized task 
families, the structure of the price-relative task ordering, and the 
extension of the framework to agentic settings.

\begin{itemize}

\item \textbf{O1 (Quality Ceiling).} Some oracle limitations are 
fundamental: an LLM without retrieval augmentation cannot recover facts 
absent from its training data; an oracle asked to decide halting will 
fail on some inputs by undecidability; an oracle trained on genuinely 
ambiguous data embeds that disagreement permanently in its response 
distribution. In each case, no token budget or query formulation can 
eliminate the error. To formalize this, define the \textbf{quality 
ceiling} of task $T$ with respect to oracle $\mathcal{O}$ as:
\[
\bar{\theta}(T, \mathcal{O}) = \sup_{\mathcal{M}} \, 
\mathbb{E}_{x \sim \mathcal{D}_X}[S(x, \mathcal{M}(x))],
\]
where the supremum is over all polynomial-time AOTMs 
$\mathcal{M} = (M, \mathcal{O})$ for task $T$. By definition, 
$\kappa_T(\theta; \alpha, \beta) = \infty$ for all $\theta > 
\bar{\theta}(T, \mathcal{O})$. The open problem is to establish 
conditions under which $\bar{\theta}(T, \mathcal{O}) < 1$ for natural 
tasks $T$. Concretely: what properties of $T$ and $\mathcal{O}$ 
determine $\bar{\theta}$? Is it computable from a description of the 
task and oracle? When does irreducible error on specific queries 
propagate to a strict ceiling for the task as a whole?

\item \textbf{O2 (Token Lower Bounds).} No non-trivial token lower 
bounds are currently known for natural AI tasks such as question 
answering, code generation, document summarization, or structured 
data extraction. Are there token lower bounds for such tasks, and 
what techniques can establish them?

\item \textbf{O3 (Frontier Universality).} The complexity frontier 
$F_T(\theta)$ is proved to be convex and upward-closed, but its 
specific shape---how steeply the tradeoff curves bend, where the 
tradeoff surface concentrates, how it scales with input size---is 
unknown for any concrete task. A natural object of study is a 
parameterized family of tasks $\{T_n\}_{n \geq 1}$, where $n$ indexes 
input size, analogous to a language family in classical complexity 
theory. Do tasks in the same family---for example, question answering, 
code generation, or summarization tasks parameterized by document 
length---share qualitative frontier shapes as $n$ grows? Is there a 
universality result where frontier shape depends only on coarse 
properties of the family---such as the type of score function, the 
oracle condition, or the input distribution---rather than on 
fine-grained task details? Such a result would make the frontier a 
practical design tool: knowing the frontier shape for one member of a 
family would immediately inform design choices for all members.

\item \textbf{O4 (Structure of the Dominance Order).} The 
price-specific dominance order $\preceq_\kappa^{(\alpha,\beta)}$ 
ranks tasks by token complexity at a fixed cost ratio $\alpha/\beta$. 
The Price-Relativity theorem shows that for tasks whose optimal AOTMs 
have different query-to-response token ratios, the ordering flips 
exactly once as $\alpha/\beta$ varies, at a threshold $K$ determined 
by the token counts. For general tasks, the relationship between 
$\kappa_{T_1}(\theta; \alpha, \beta)$ and $\kappa_{T_2}(\theta; 
\alpha, \beta)$ as a function of $\alpha/\beta$ may be more complex. 
Can the ordering flip multiple times as $\alpha/\beta$ varies? What 
properties of $T_1$ and $T_2$ determine the number of flips and the 
location of the thresholds? And what is the global structure of the 
family of dominance orders over all cost ratios: does it have finite 
complexity, or can arbitrarily intricate ranking patterns arise?

\item \textbf{O5 (Agentic Extension).} Extend the framework to agentic 
tasks where $M$ interacts with an external environment---issuing actions 
such as API calls, file writes, or web navigation---and receives 
observations whose outcomes shape subsequent oracle queries. The 
standard task structure $T = (X, Y, S, \mathcal{D}_X)$ must be 
generalized to account for the action-observation loop, and the 
complexity frontier becomes a reward-cost surface trading off token 
expenditure against cumulative reward. Does this surface retain 
convexity and the other structural properties proved here, or does the 
sequential dependence of actions on observations introduce fundamentally 
new complexity?

\end{itemize}

\section{Conclusion}\label{Conclusion}

We have introduced token complexity as a formal resource measure for 
AI-augmented computing and developed it as a computational theory of 
measuring the cost of using oracles. The AI-Oracle Turing Machine 
model formalizes the boundary between probabilistic computation and 
stochastic oracle access, treating token cost as the fundamental 
measure of oracle calls. Token complexity and the complexity frontier 
together provide not just a vocabulary for measuring oracle cost but 
a principled basis for reasoning about when and how to reduce it.

The results are complementary: the token complexity theorems 
characterize the cost structure of individual tasks, while the 
frontier theorems characterize the geometry of resource tradeoffs 
across all feasible designs. Together they show that AI-augmented 
computing is not just a new programming paradigm but a new complexity 
regime, one where the cost of stochastic oracle interaction introduces 
structure that classical time and space complexity cannot capture.

These results match experience from deployed AI systems: quality 
improvements consistently require more careful token use, small cost 
changes predictably shift cost-optimal designs, and the 
diminishing-returns pattern of the tradeoff surface is observed 
whenever engineers push toward the limits of token optimization.

Token complexity extends classical complexity theory by adding a 
third axis to the time-space plane, with both descriptive and 
prescriptive power. Descriptively, it characterizes the cost 
structure of AI-augmented computing. Prescriptively, it tells 
designers when to reorganize queries, when to maintain state across 
oracle calls, when to decompose tasks differently under different 
costs, and when to accept that no further optimization is possible. 
As AI-augmented computing becomes the norm rather than the 
exception, this third axis will increasingly determine what is 
\textit{feasible}, what is \textit{efficient}, and what is 
\textit{cheap}.

\section*{Acknowledgment}
I am grateful to Allan Guo and Frank Sun at Librum Technologies 
for insightful conversations on token use and costs in 
practice, which helped shape my thinking on token complexity.

\bibliographystyle{plain}
\bibliography{reference}

@book{HomerSelman2011,
  author    = {Homer, Steven and Selman, Alan L.},
  title     = {Computability and Complexity Theory},
  edition   = {2nd},
  publisher = {Springer},
  year      = {2011}
}

@book{DuKo2001,
  author    = {Du, Ding-Zhu and Ko, Ker-I},
  title     = {Problem Solving in Automata, Languages, and Complexity},
  publisher = {John Wiley \& Sons},
  year      = {2001}
}

@book{AroraBarak2009,
  author    = {Arora, Sanjeev and Barak, Boaz},
  title     = {Computational Complexity: A Modern Approach},
  publisher = {Cambridge University Press},
  year      = {2009}
}

@article{HopcroftPaulValiant1977,
  author    = {Hopcroft, John E. and Paul, Wolfgang J. and Valiant, Leslie G.},
  title     = {On Time versus Space},
  journal   = {Journal of the ACM},
  volume    = {24},
  number    = {2},
  pages     = {332--337},
  year      = {1977}
}

@article{Levin1986,
  author  = {Levin, Leonid A.},
  journal = {SIAM Journal on Computing},
  number  = {1},
  pages   = {285--286},
  title   = {Average case complete problems},
  volume  = {15},
  year    = {1986}
}

@inproceedings{Sennrich2016,
  author    = {Sennrich, Rico and Haddow, Barry and Birch, Alexandra},
  booktitle = {The 54th Annual Meeting of the Association for Computational Linguistics},
  pages     = {1715--1725},
  publisher = {ACL},
  title     = {Neural machine translation of rare words with subword units},
  year      = {2016}
}

@article{Kolmogorov1965,
  author  = {Kolmogorov, Andrei Nikolaevich},
  journal = {Problems of Information Transmission},
  number  = {1},
  pages   = {1--7},
  title   = {Three approaches to the quantitative definition of information},
  volume  = {1},
  year    = {1965}
}

@article{Shannon1959,
  author  = {Shannon, Claude Elwood},
  journal = {IRE National Convention Record},
  pages   = {142--163},
  title   = {Coding theorems for a discrete source with a fidelity criterion},
  volume  = {7},
  year    = {1959}
}

@article{Gill1977,
  author  = {Gill, John},
  journal = {SIAM Journal on Computing},
  number  = {4},
  pages   = {675--695},
  title   = {Computational complexity of probabilistic Turing machines},
  volume  = {6},
  year    = {1977}
}

@incollection{Wang1997,
  author    = {Wang, Jie},
  booktitle = {Complexity Theory Retrospective II},
  editor    = {Hemaspaandra, Lane and Selman, Alan},
  pages     = {295--328},
  publisher = {Springer},
  title     = {Average-case computational complexity theory},
  year      = {1997}
}

@article{Wang2025,
  author  = {Wang, Jie},
  journal = {{AI} Matters},
  number  = {3},
  pages   = {8--11},
  title   = {{AI}-oracle machines for intelligent computing},
  volume  = {10},
  year    = {2025}
}

@book{Li2008,
  author    = {Li, Ming and Vit\'{a}nyi, Paul},
  publisher = {Springer},
  title     = {An Introduction to Kolmogorov Complexity and Its Applications. 3rd ed.},
  year      = {2008}
}

@inproceedings{Schuster2012,
  author    = {Schuster, Mike and Nakajima, Kaisuke},
  booktitle = {IEEE International Conference on Acoustics, Speech and Signal Processing},
  pages     = {5149--5152},
  publisher = {IEEE},
  title     = {Japanese and {K}orean voice search},
  year      = {2012}
}

@inproceedings{Kudo2018,
  author    = {Kudo, Taku and Richardson, John},
  booktitle = {The 2018 Conference on Empirical Methods in Natural Language Processing: System Demonstrations},
  pages     = {66--71},
  publisher = {ACL},
  title     = {SentencePiece: A simple and language independent subword tokenizer and detokenizer for neural text processing},
  year      = {2018}
}

@book{Cover2006,
  author    = {Cover, Thomas M. and Thomas, Joy A.},
  publisher = {Wiley},
  title     = {Elements of Information Theory. 2nd ed.},
  year      = {2006}
}

\end{document}